\documentclass{article}

\usepackage{arxiv}
\usepackage{color}
\usepackage{overpic}
\usepackage{subfig}
\usepackage{hyperref}
\usepackage{graphicx}      
\usepackage{amsmath,amssymb,amsfonts, amsthm}
\newtheorem{lemma}{Lemma}
\newtheorem{assum}{Assumption}
\newtheorem{theorem}{Theorem}

\newtheorem{remark}{Remark}
\newtheorem{problem}{Problem}
\newtheorem{subproblem}{Subproblem}
\newtheorem{definition}{Definition}

\newcommand{\cl}[2][black]{\textcolor{#1}{#2}}
 
\title{\LARGE \bf

Rigidity-Aware Formation Tracking under Sensing Range Constraints via Single Control Barrier Function Constraint\\

}

\author{ \href{https://orcid.org/0000-0003-0212-0762}{\includegraphics[scale=0.06]{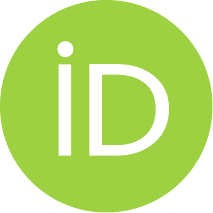}\hspace{1mm}S. Saharsh} \\
	Department of Cyber Physical Systems\\
	Indian Institute of Science\\
	Bengaluru, 560012, India \\
	\texttt{saharsh2021@iisc.ac.in} \\
        \And
	Pushpak Jagtap \\
	Department of Cyber Physical Systems\\
	Indian Institute of Science\\
	Bengaluru, 560012, India\\
	\texttt{pushpak@iisc.ac.in} \\
}

\renewcommand{\shorttitle}{\textit{arXiv} version}
\begin{document}

\maketitle

\begin{abstract}

This paper presents a control framework for formation tracking and rigidity maintenance in heterogeneous multi-robot systems with nonlinear dynamics under sensing range constraints. Since formation tracking alone does not ensure rigidity maintenance with a limited sensing range, despite rigidity being a prerequisite for establishing and preserving a unique formation, our work integrates both objectives through a single Control Barrier Function (CBF)-like constraint within a quadratic optimization framework. The proposed distributed controller requires only local relative information from neighbors, as verified with simulation case studies. 

\end{abstract}

\section{Introduction}
Multi-robot systems (MRS) have emerged as a powerful paradigm for accomplishing complex tasks through cooperation, offering improved robustness, flexibility, and scalability compared with single-robot solutions. In GPS-denied environments such as indoor settings, underwater domains, or deep-space missions, agents must rely on local neighbor based information for coordination and control \cite{franchi2012modeling}. The analysis and stabilization of group formations in MRS, based solely on inter-agent measurements, needs rigid communication graph \cite{zelazo2015decentralized}. Bearing-based formation control, with rigidity formalized through the bearing Laplacian matrix, has consequently emerged as an effective approach for achieving formation objectives under such minimal sensing requirements \cite{zhao2015bearing}.

Lyapunov-based bearing-only controllers \cite{pampatwar2021planar} achieve formation tracking but exhibit limited scalability and adaptability with time-varying communication graph. Quadratic-programming (QP)-based controllers address this limitation by enabling real-time constraint handling, though conflicts between nominal formation objectives and imposed constraints can render the resulting QP infeasible \cite{wu2023quadratic}. Despite significant progress in formation control, most existing works \cite{pampatwar2021planar, zhao2016localizability, rayabagi2024formation} assume that a rigid communication graph is already available and maintained throughout the mission. However, in practical scenarios, agents are often subject to limited sensing ranges \cite{bhatia2025decentralized} and restricted fields of view \cite{dias2016distributed}, making rigidity maintenance itself a challenging problem. Motivated by these limitations, this paper proposes a control-barrier-function (CBF)-only framework that unifies formation tracking and rigidity maintenance, ensuring the network remains rigid even if agents have limited sensing range.



\cl{This paper extends the preliminary results presented in \cite{saharsh2026cbfformation} for
formation tracking in heterogeneous multi-agent systems using the CBF-QP framework, where a
reachability control barrier function (reach-CBF) was introduced to guarantee safe-set invariance
and finite-time target-set reachability within a single-constraint CBF-QP framework, assuming the
communication graph remains rigid for all time. Here, we adopt this reach-CBF-based framework on a
{time-varying}, limited-sensing-range graph, removing the assumption of permanent rigidity and
instead treating rigidity preservation as an explicit control objective. We provide a comprehensive and comparative
case study illustrating the importance of rigidity maintenance for formation tracking, where
communication links are broken or created depending on limited sensing range. \\
Our main contribution is a unified formation tracking and rigidity maintenance strategy for
heterogeneous nonlinear multi-agent systems under limited sensing range, achieved within the single reach-CBF constraint-based quadratic program, using only local relative information
from neighbors.}
\section{Problem Formulation}
\label{sec: prelim and prob defn}
\textit{Notations}. Scalars are denoted as $x \in \mathbb{R}$, and vectors as $\boldsymbol{x}$ (bold letter).  Let $\mathcal{S}$ be the unit circle. 
The closed ball $\mathbb{B}(\boldsymbol{c},\lambda) \subseteq \mathbb{R}^2$ is defined as $\mathbb{B}(\boldsymbol{c}, \lambda) := \{ \boldsymbol{x} \in \mathbb{R}^2 : \|\boldsymbol{x} - \boldsymbol{c}\| \leq \lambda\}$, where $\boldsymbol{c}$ is the center and $\lambda \in \mathbb{R}^{+}$ is the radius. 
For $a, b \in \mathbb{N}$ and $a < b$, $[a,b]_{\mathbb{N}}$ denotes the closed interval in $\mathbb{N}$. All other notation follows standard mathematical conventions.
\subsection{System Description}
We consider a heterogeneous multi-agent system (MAS) consisting of $n$ agents, where each agent is indexed by $i \in [1,\ n]_{\mathbb{N}}$. Each $i^{th}$ agent (modeled as a planar robot) is represented as:
\begin{align}
        \dot{\boldsymbol{p}}_i = \boldsymbol{v}_i, 
        \dot{\boldsymbol{v}}_i = f_{v_i}(\boldsymbol{v}_i) + h_{v_i}(\boldsymbol{v}_i)\boldsymbol{u}_i, \label{sys}
\end{align}
    where $\boldsymbol{p}_i,\boldsymbol{v}_i\in \mathbb{R}^2$ and $\boldsymbol{u}_i\in \mathbb{R}^2$ are the position in the 2D plane, velocity, and input vectors, respectively. The maps $f_{v_i} : \mathbb{R}^2 \times \mathbb{R}^2 \rightarrow \mathbb{R}^2$, $h_{v_i}: \mathbb{R}^2 \times \mathbb{R}^2 \rightarrow \mathbb{R}^{2 \times 2}$ are globally Lipschitz continuous functions with Lipschitz constants $\Gamma_{if}, \Gamma_{ih}$, respectively. \cl{All agents are assumed to share a common orientation reference frame.}

\begin{remark}
    A broad class of planar robotic systems, including simple double-integrator, Ackermann-drive vehicles, differential-drive robots, aerial vehicles in altitude-hold mode, can be represented by \eqref{sys}. Refer \cite{sawarkar2026sliding}[II-B] for derivations.
\end{remark}

\subsection{Graph-based Communication Framework} 
Given a MAS of $n$ agents, its communication topology is modeled by an undirected graph $\mathcal{G}(t)=(\mathcal{V},\mathcal{E}(t))$, where $\mathcal{V}=[1,n]_{\mathbb{N}}$ denotes the set of agents (modeled as \eqref{sys}) and $\mathcal{E}(t)\subseteq\mathcal{V}\times\mathcal{V}$ denotes the set of edges at time $t \in \mathbb{R}_{0}^{+}$, with $(i,j)\in\mathcal{E}(t)$ indicating information exchange between agents $i$ and $j$. We consider that all the agents have the same sensing range $\lambda \in \mathbb{R}^{+}$. The edge condition $(i,j) \in \mathcal{E}(t) \iff \boldsymbol{p}_j(t) \in \mathbb{B}(\boldsymbol{p}_i(t), \lambda) \iff \boldsymbol{p}_i(t) \in \mathbb{B}(\boldsymbol{p}_j(t), \lambda) \iff (j,i) \in \mathcal{E}(t)$, favoring the undirected graph topology. The neighbors of the $i^{th}$ agent are denoted as $\mathcal{N}_{i}$, where $\mathcal{N}_{i}(t) := \{j \in \mathcal{V}: (i,j) \in \mathcal{E}(t)\}$. For a given edge $(i,j) \in \mathcal{E}(t)$, we denote the relative position between agent $i$ and $j$ by $\boldsymbol{p}_{ij} := \boldsymbol{p}_{j} - \boldsymbol{p}_{i}$ and the bearing by $\boldsymbol{g}_{ij} := \frac{\boldsymbol{p}_{ij}}{\|\boldsymbol{p}_{ij}\|}, \forall (i,j) \in \mathcal{E}(t)$.
The bearing configuration $\boldsymbol{g} \in \mathbb{R}^{2|\mathcal{E}|}$ is formed by stacking the bearing vectors $\boldsymbol{g}_{ij}$ for all edges $(i,j) \in \mathcal{E}(t)$.
Given $\boldsymbol{g}_{ij} \in \mathbb{R}^2$, orthogonal projection matrix $P_{\boldsymbol{g}_{ij}} \in \mathbb{R}^{2 \times 2}$ is defined as $P_{\boldsymbol{g}_{ij}} = I_2 \ - \ \boldsymbol{g}_{ij} \boldsymbol{g}_{ij}^{\top}$.
Using the matrix $P_{\boldsymbol{g}_{ij}}$ and $\mathcal{G}(t)$, the bearing Laplacian matrix $\mathcal{B}(t) \in \mathbb{R}^{2n \times 2n}, n = |\mathcal{V}|$ is defined as
\begin{align}\label{bear}
    [\mathcal{B}(t)]_{ij} =
\begin{cases}
{0}_{2}, & i \ne j,\ (i,j) \notin \mathcal{E}(t), \\[8pt]
- P_{\boldsymbol{g}_{ij}}, & i \ne j,\ (i,j) \in \mathcal{E}(t), \\[8pt]
\sum\limits_{k \in \mathcal{N}_i} P_{\boldsymbol{g}_{ik}}, & i = j,\ i \in \mathcal{V}.
\end{cases}\end{align}
The matrix $\mathcal{B}(t)$ is used later in Subsection \ref{subsec:problem formulation} to define the bearing rigidity. 
\cl{We adopt a leader--follower framework, where agents with access to global position information (e.g., via GPS) are designated leaders $\mathcal{V}_L \subset \mathcal{V}$, $n_l = |\mathcal{V}_L|$, and the remaining agents, relying only on local relative measurements, are followers $\mathcal{V}_F = \mathcal{V}\setminus\mathcal{V}_L$, $n_f = n - n_l$.}
The matrix $\mathcal{B}(t)$ for the leader-follower framework can be partitioned into block matrices as 
    $\mathcal{B}(t) = \begin{bmatrix}
\mathcal{B}_{ll}(t) & \mathcal{B}_{lf}(t) \\
\mathcal{B}_{fl}(t) & \mathcal{B}_{ff}(t)
\end{bmatrix}$,
where $\mathcal{B}_{ll} \in \mathbb{R}^{2n_l \times 2n_l}$, $\mathcal{B}_{lf} \in \mathbb{R}^{2n_l \times 2n_f}$,
$\mathcal{B}_{fl} \in \mathbb{R}^{2n_f \times 2n_l}$, $ \mathcal{B}_{ff} \in \mathbb{R}^{2n_f \times 2n_f}$, where $l$ and $f$ represent leaders and followers, respectively.
\subsection{Problem Formulation}
\label{subsec:problem formulation}
The MAS of $n$ agents, connected over a communication graph $\mathcal{G}(t)$ with the leader-follower framework, is required to reach, maintain, and track the desired formation, given by the desired leader positions $\boldsymbol{p}^{d}_{l}, \forall l \in \mathcal{V}_{L}$ and the desired bearing vectors $\boldsymbol{g}_{ij}^{d}$, while preserving rigidity under sensing range $\lambda$. The unique formation of the MAS is characterized by a desired position configuration vector $\boldsymbol{p}^{d} \in \mathbb{R}^{2n}$:
\begin{align}\label{dpos}
   \boldsymbol{p}^{d} := [(\boldsymbol{p}_{1}^{d})^{\top}, \cdots,(\boldsymbol{p}_{n}^{d})^{\top}]^{\top}, 
\end{align}
where $\boldsymbol{p}_{i}^{d}$, $i\in\mathcal{V}$ denote the desired formation position of agent $i$. 
The formation objective is defined below.
\begin{definition} [Unique formation]
    Given a desired position configuration $\boldsymbol{p}^{d}$, with agents (modeled as \eqref{sys}) in MAS, connected as $\mathcal{G}(t)$ under a leader-follower framework, it is said to achieve a \textit{unique formation} if the following holds: $\lim_{t \rightarrow \infty} \| \boldsymbol{p}_i(t) - \boldsymbol{p}_i^{d}\| = 0, \forall i \in \mathcal{V}.$
    \label{uform}
\end{definition}
With the given leader positions $\boldsymbol{p}_l^d, l \in \mathcal{V}_L$, the desired formation in the bearing configuration $\boldsymbol{g}^d$, denoted as $\mathcal{F}(\boldsymbol{g})$, described by the desired bearing vectors $\boldsymbol{g}_{ij}^d := \frac{\boldsymbol{p}_{ij}^{d}}{\|\boldsymbol{p}_{ij}^{d}\|}, \forall (i,j) \in \mathcal{E}(t)$, $\boldsymbol{p}_{ij}^{d}= \boldsymbol{p}_{j}^{d} - \boldsymbol{p}_{i}^{d}$, is achieved if $\lim_{t \rightarrow \infty} \langle\boldsymbol{g}_{ij}(t),\boldsymbol{g}_{ij}^{d}\rangle = 1, \forall (i,j) \in \mathcal{E}(t). $
\begin{assum}[\cite{zhao2016localizability}]
     The initial graph $\mathcal{G}(0)$ is infinitesimally bearing rigid, i.e., $det(\mathcal{B}_{ff}(0)) \neq 0$ with \cl{initial leaders only subgraph $\mathcal{G}_{\mathcal{V}_L}(0):= (\mathcal{V}_L, \mathcal{E}_{L}(0)), \mathcal{E}_{L}(0) = \{(i,j) : i,j \in \mathcal{V}_L, (i,j) \in \mathcal{E}(0)\}$ is also bearing rigid.} 
    \label{A1}
\end{assum}
In contrast to conventional approaches \cite{zhao2015bearing, zhao2019bearing}, which assume bearing rigidity for all time, here bearing rigidity is assumed only initially. Rigidity preservation is treated as a rigidity maintenance objective and established as follows.
\begin{definition}[Rigidity maintenance]\label{def_connectivity}
Consider a MAS connected through a communication graph
$\mathcal{G}(t)=(\mathcal{V},\mathcal{E}(t))$ with agent's sensing range $\lambda \in \mathbb{R}^{+}$. The system is said to maintain rigidity if the communication links
preserve bearing rigidity for all time, i.e., $det(\mathcal{B}_{ff}(t)) \neq 0, \forall t \in \mathbb{R}_{0}^{+}$.
\end{definition}
The transition from a formation $\mathcal{F}(\boldsymbol{g})$ defined in bearing to the unique formation is established next.
\begin{lemma}[\cite{zhao2016localizability}]
    Given a MAS in the leader-follower framework, connected as $\mathcal{G}(t)$, distinct leaders positions as $\boldsymbol{p}_{l}^{d}$, the desired formation is uniquely determined with $\mathcal{F}(\boldsymbol{g})$, i.e., $\forall (i,j) \in \mathcal{E}(t), \lim_{t \rightarrow \infty} \boldsymbol{g}_{ij}(t) = \boldsymbol{g}_{ij}^{d} \Rightarrow \forall i \in \mathcal{V}_{F}, \lim_{t \rightarrow \infty} \boldsymbol{p}_{i}(t) = \boldsymbol{p}_{i}^{d}$, if and only if $\mathcal{B}_{ff}(t)$ is non-singular $\forall t \in \mathbb{R}_{0}^{+}$. In addition, if $\mathcal{B}_{ff}(t)$ is non-singular, then the system has at least two leaders (i.e., $n_l \geq 2$). 
    \label{lem1}
\end{lemma}
\cl{Note that we need two leaders since one alone cannot uniquely resolve agent's position from bearing.}
\begin{definition} [Rigid formation]\label{mform}
   Consider a MAS connected as $\mathcal{G}(t)$ under a leader-follower framework, with preserved rigidity as in Definition \ref{def_connectivity}, with time-varying leader trajectories $\boldsymbol{p}_l^{d}(t)$ moving with constant velocity $\boldsymbol{v}_{l}$ and the desired bearing configuration $\boldsymbol{g}^{d}$. With Assumption \ref{A1}, the agents are said to achieve a unique \textit{rigid formation} $\mathcal{F}(\boldsymbol{g}(t))$ if $
    \lim_{t \rightarrow \infty} \| \boldsymbol{g}(t)- \boldsymbol{g}^{d}\| = 0$.
\end{definition}
\begin{problem}\label{mainprob}
    Given a heterogeneous MAS with agent dynamics as in \eqref{sys}, a communication graph $\mathcal{G}(t)$ with leader-follower framework under Assumption \ref{A1} and limited sensing range $\lambda$, a desired unique rigid formation $\mathcal{F}(\boldsymbol{g}(t))$ specified by the leader's trajectory $\boldsymbol{p}_{l}^{d}(t)$ and desired bearing vectors $\boldsymbol{g}_{ij}^{d}$, design a distributed control strategy using bearing information $\boldsymbol{g}_{ij}$, so that it achieves a unique rigid formation $\mathcal{F}(\boldsymbol{g}(t))$ (as in Definition \ref{mform}) under bearing rigidity maintenance (as in Definition \ref{def_connectivity}).
\end{problem}
We solve Problem \ref{mainprob} via a time-varying formation tracking controller using the reachability CBF constraint in a QP framework, as described in \cite[Section 2.4]{saharsh2026cbfformation}.
\subsection{Reachability Control Barrier Function}
We introduce a CBF-like function that ensures forward invariance of a safe set $X \subset \mathbb{R}^{n}$ as well as reachability of a target set $R \subset X$, in finite time $\tau \in \mathbb{R}^{+}$. We employ a continuously differentiable function $b: \mathbb{R}^{n} \rightarrow \mathbb{R}$ in order to define the safe set $X$, where $X:= \{{\boldsymbol{x}} \in \mathbb{R}^{n}: b(\boldsymbol{x}) \geq 0\},$
and the boundary set $\partial X := \{{\boldsymbol{x}} \in \mathbb{R}^{n}: b(\boldsymbol{x}) = 0 \}$. Using the function $b$, which has unique maxima in $X$, $\max_{\boldsymbol{x} \in X} b(\boldsymbol{x}) = M$, with constant $\mu \leq M, \mu\in \mathbb{R}^{+}_{0}$, the target set $R$ is defined as $R:= \{{\boldsymbol{x}} \in \mathbb{R}^{n}: b(\boldsymbol{x}) \geq \mu\}.$
Now we define our CBF-like function, using $b$ as below.
\begin{definition}[Reachability CBF]\label{cbfdef}
    Given a compact set $X \subset \mathbb{R}^n$, target set $R$, using a continuously differentiable function $b$, with unique maxima in $X$, associated with the affine control system $\dot{\boldsymbol{x}} = f({\boldsymbol{x}}) + h({\boldsymbol{x}})\boldsymbol{u}$, where $f: \mathbb{R}^{n} \rightarrow \mathbb{R}^{n}$, $h: \mathbb{R}^{n} \rightarrow \mathbb{R}^{n \times m}$,  $\boldsymbol{u} \in \mathbb{R}^{m}$, then $b$ is a \emph{reachability control barrier function} (reach CBF) if for reachability time $ \tau \in \mathbb{R}^{+}$, it holds that $\forall \boldsymbol{x} \in X$, 
    \begin{align}
        \label{rcbfcons}
        &\sup_{{\boldsymbol{u}} \in \mathbb{R}^m}\{\underbrace{\nabla b(\boldsymbol{x}) \cdot f(\boldsymbol{x})}_{K_f(\boldsymbol{x})} + \underbrace{\nabla b(\boldsymbol{x}) \cdot h(\boldsymbol{x})}_{K_h^{\top}(\boldsymbol{x})}{\boldsymbol{u}}\} \geq \frac{\mu}{\tau},
    \end{align}
where $K_f : \mathbb{R}^n \rightarrow \mathbb{R}$ and $K_h : \mathbb{R}^n \rightarrow \mathbb{R}^{m}$ are locally Lipschitz continuous functions. 
\end{definition}
Next, the CBF-QP-like optimization problem is stated below.
\begin{theorem}\label{thm1}
Given a compact set $X$, a target set $R$, using the reach CBF $b$, as in Definition \ref{cbfdef}, if $\frac{\partial b}{\partial {\boldsymbol{x}}} \neq \boldsymbol{0}, \forall {\boldsymbol{x}} \in \partial X$, for initial state $\boldsymbol{x}(0) \in X$, and the control input ${\boldsymbol{u}}$ is computed based on the optimization problem as
\begin{align}\label{QPdef}
    &{\boldsymbol{u}}({\boldsymbol{x}}) = \arg \min_{\boldsymbol{q} \in \mathbb{R}^{m}} \quad \frac{1}{2}\boldsymbol{q}^{\top}\boldsymbol{q} , \nonumber \\
    &s.t. \quad K_{f}(\boldsymbol{x}) + K_{h}^{\top}(\boldsymbol{x})\boldsymbol{q} \geq \frac{\mu}{\tau}, 
\end{align}
where $\tau \in \mathbb{R}^{+}, 0 \leq \mu \leq M$, then the set $R$ is said to be reachable in finite time $\tau$, inside the forward invariant set $X$.    
\end{theorem}
\begin{proof}
    By using the initial condition $b(\boldsymbol{x}(0)) \geq 0$, and after integrating the constraint in \eqref{QPdef}, we get $b(\boldsymbol{x}(\tau)) - b(\boldsymbol{x}(0)) \geq \int_{0}^{\tau}\frac{\mu}{\tau}ds  = \mu \ge 0.$
Thus, the set $R$ is reached in time $\tau$. The compact set $X$ is forward invariant using (\cite{ames2019control}[Theorem 1]) \cl{since the reach-CBF constraint keeps $b(x)$ non-decreasing}, as $b(\boldsymbol{x}(0)) \geq 0, \dot{b} \geq 0, \forall t > 0$.
\end{proof}
Further, our main Problem \ref{mainprob} is decomposed into subproblems for each follower agent $i$ as discussed in the following section.

\section{Subproblem Formulation}
\label{sec: distributed execution}
To solve Problem~\ref{mainprob} in a distributed manner, we decompose the global formation task $\mathcal{F}(\boldsymbol{g}(t))$ and rigidity maintenance task into local subtasks, defined over each agent's neighborhood $\mathcal{N}_i(t)$ formed on the sensing ball $\mathbb{B}(\boldsymbol{p}_i, \lambda)$. To address Problem~\ref{mainprob} using only local information, we assume:
\begin{assum}
    The relative position $\boldsymbol{p}_{ij}(t)$, velocity $\boldsymbol{v}_{ij}(t)$ of $j^{th}$ agent with respect to $i^{th}$ agent, along with the control input $\boldsymbol{u}_j(t)$ taken by agent $j$, can be measured at time $t$, If they are connected, i.e., if $j \in \mathcal{N}_i(t)$.
    \label{A2}
\end{assum}
\begin{figure}[thpb]
      \begin{center}
      \includegraphics[scale=0.2]{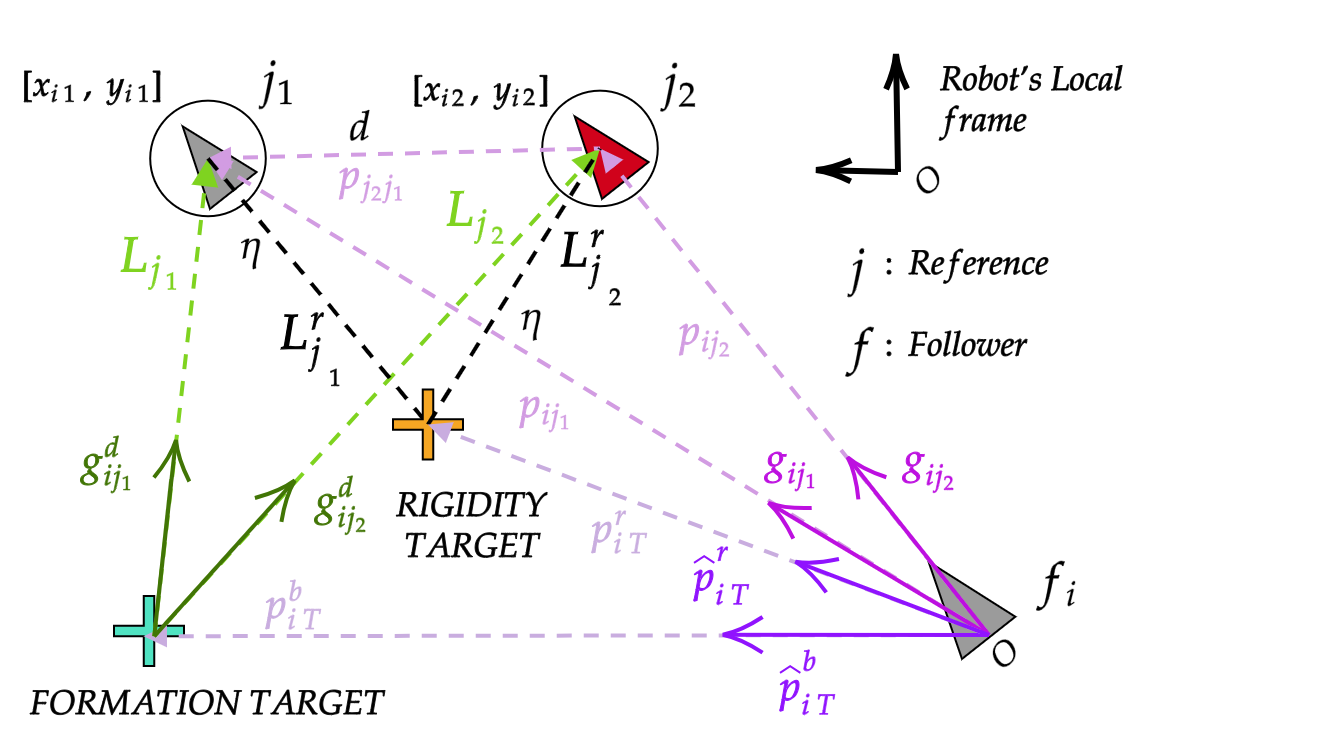}
      \caption{Formation control problem for follower agent $i$.}
      \label{probset}
      \end{center}
\vspace{-5mm}
\end{figure}
\begin{lemma}[Leader-Rooted Reference Assignment]
\label{lem2}
\cl{Given the MAS connected as $\mathcal{G}(0)$ under Assumption~\ref{A1},
there exists a level assignment $\ell:\mathcal{V}\to\mathbb{N}_0$, with
$\ell(l)=0,\ \forall l\in\mathcal{V}_L$, and, for every follower
$i\in\mathcal{V}_F$, there exists at least two neighbors $j_1,j_2\in\mathcal{N}_i(0)$
with $\ell(j_1),\ell(j_2) < \ell(i)$ and non-parallel desired bearings
$\boldsymbol{g}^d_{ij_1}\not\parallel\boldsymbol{g}^d_{ij_2}$.}
\end{lemma}

\begin{proof}
\cl{
Since $\det(B_{ff}(0))\neq 0$ by Assumption~\ref{A1}, $\mathcal{G}(0)$ is infinitesimally bearing rigid, hence also distance rigid \cite[Theorem~8]{zhao2015bearing}, and $\mathcal{G}_{\mathcal{V}_L}(0)$ is itself rigid. Starting from the rigid seed $\mathcal{G}_{\mathcal{V}_L}(0)$, removing redundant edges from $\mathcal{G}(0)$ one at a time while preserving bearing rigidity yields a minimally bearing-rigid subgraph $\hat{\mathcal{G}}(0)\subseteq\mathcal{G}(0)$. In 2D, $\hat{\mathcal{G}}(0)$ admits a Henneberg-type vertex-addition sequence \cite[Section~2.4]{anderson2003operations} rooted at $\mathcal{G}_{\mathcal{V}_L}(0)$: each follower $i$ is added via exactly two edges to already-placed vertices $j_1,j_2$. This gives only the existence of such $j_1,j_2$, with non-parallel bearings. Assigning $\ell(l)=0$ for leaders and $\ell(i)=1+\max\{\ell(j_1),\ell(j_2)\}$ along the sequence gives $\ell(j_1),\ell(j_2)<\ell(i)$ by construction.
Suppose, for contradiction, $\boldsymbol{g}^d_{ij_1}\parallel\boldsymbol{g}^d_{ij_2}$ for any pair $j_1,j_2$. Then $\boldsymbol{p}_i^d$ cannot be uniquely determined from $j_1,j_2$ by the ``only if'' part of Lemma~\ref{lem1}, implying $\det(B_{ff}(0))=0$, a contradiction. 
Hence $\ell$ satisfies all required properties.
}
\end{proof}
\begin{remark}[Selection Rule]
\cl{For each follower $i \in \mathcal{V}_F$, the reference pair $j_1, j_2$ in Lemma \ref{lem2} is selected as the two lowest-level, non-parallel bearing neighbors, i.e.,
$\{j_1, j_2\} = \operatorname*{arg\,min}_{\{j,k\} \subseteq \mathcal{N}_i(0),\ \boldsymbol{g}^d_{ij} \not\parallel \boldsymbol{g}^d_{ik}} \big(\ell(j), \ell(k)\big),$
under the construction order $\ell$ which is obtained by the sequential Henneberg construction and computed offline once $\mathcal{G}(0)$ is given. }
\end{remark}

\cl{We define the time-invariant reference subgraph $\hat{\mathcal{G}} := (\mathcal{V}, \hat{\mathcal{E}}(0))$, $\hat{\mathcal{E}}(0) = \{(i, j_1),(i, j_2): i \in \mathcal{V}_F\} \cup \mathcal{E}_L(0)$, to contain only
the two fixed reference edges assigned to each follower
(Lemma~\ref{lem2}), with follower-follower bearing
Laplacian block $\hat{\mathcal{B}}_{ff}(t)$. Lemma~\ref{lem2}
establishes that this subgraph is rigid at $t = 0$, via Henneberg vertex-addition sequence\cite[Theorem 7]{anderson2003operations} on the rigid subgraph of only leaders $\mathcal{G}_{\mathcal{V}_L}(0)$. Moreover, since $\hat{\mathcal{E}}(t)$
contains exactly two edges per follower with no redundancy,
losing either reference edge strictly breaks this rigidity, by the
``only if'' part of Lemma~\ref{lem1}, giving criticality, unlike
the fully sensed graph $\mathcal{G}(t)$. Building on this, if every follower maintains its assigned reference pair for
all time, $\hat{\mathcal{B}}_{ff}(t)$ remains nonsingular for all $t$;
and since $\hat{\mathcal{E}}(t) \subseteq \mathcal{E}(t)$ with
bearing Laplacian matrix rank non-decreasing under edge addition, this in turn
forces $\det(\mathcal{B}_{ff}(t)) \neq 0$. Thus, maintaining only the
sparse, locally-assigned reference edges is sufficient to maintain
rigidity of the full sensed graph $\mathcal{G}(t)$.}
\begin{subproblem}
Consider a formation control problem for a follower agent $i \in \mathcal{V}_{F}$ with respect to its own local frame of reference. Given the relative position of the reference nodes $j_1$ and $j_2$ as $\boldsymbol{p}_{ij_1} = [x_{i1} ,\ y_{i1}]^{\top}$ and $\boldsymbol{p}_{ij_2} = [x_{i2},\ y_{i2}]^{\top}$, respectively. By leveraging Assumption \ref{A2}, we need to achieve \textit{formation tracking}, i.e., the agent $i$ aims to satisfy the formation subtask $\mathcal{F}_i(\boldsymbol{g}(t)) :=   \lim_{t \rightarrow \infty}\|\boldsymbol{g}_{ij}(t) - \boldsymbol{g}_{ij}^{d}\| = 0, \forall j \in \{j_1,j_2\},$ and maintain \textit{local bearing rigidity constraint} with references as $\boldsymbol{p}_{j_1}(t), \boldsymbol{p}_{j_2}(t) \in \mathbb{B}(\boldsymbol{p}_i(t), \lambda)$ or $\max\{\|\boldsymbol{p}_{ij_1}(t)\|,\ \|\boldsymbol{p}_{ij_2}(t)\|\} \leq \lambda, \forall t \in \mathbb{R}_{0}^{+}$, \cl{together with non-collinearity for reference links as $\sin^{-1} (\boldsymbol{g}_{ij_1}(t)) \times \boldsymbol{g}_{ij_2}(t)) \neq 0, , \forall t \in \mathbb{R}_{0}^{+}$}.
\label{sp}
\end{subproblem} 
Next, we show that solving the Subproblem \ref{sp} for each follower agent $i\in\mathcal{V}_F$ also solves the Problem \ref{mainprob}.
\begin{theorem}\label{theo1}
Given a heterogeneous MAS, connected as $\mathcal{G}(t)$ with a leader-follower framework, satisfying Assumptions \ref{A1} and \ref{A2} with sensing range $\lambda$, and if each follower agent $i$ achieves the formation subtask $\mathcal{F}_i(\boldsymbol{g}(t))$ such that for reference nodes $j_1, j_2$, we have $\boldsymbol{p}_{j_1}(t), \boldsymbol{p}_{j_2}(t) \in \mathbb{B}(\boldsymbol{p}_i(t), \lambda), \forall t \in \mathbb{R}_{0}^{+}$, then the MAS solves Problem \ref{mainprob}.
\end{theorem}
\begin{proof}
   Given a unique rigid formation $\mathcal{F}(\boldsymbol{g}(t))$, described by $ \boldsymbol{g}^{d}$, formation subtask $\mathcal{F}_i(\boldsymbol{g}(t))$ implies a unique local formation for each follower agent $i$, with respect to their reference nodes, as the local bearings $\boldsymbol{g}_{ij}(t) \rightarrow \boldsymbol{g}_{ij}^{d}, \forall j \in \{j_1,j_2\}$. Therefore, with the leaders moving as predefined and each follower satisfying $\mathcal{F}_i(\boldsymbol{g}(t))$, the collective satisfaction of all formation subtasks, i.e., $\{\mathcal{F}_i(\boldsymbol{g}(t))\}_{{\forall i\in \mathcal{V}_{F}}}$, implies that $\boldsymbol{g}_{ij}(t) \rightarrow \boldsymbol{g}_{ij}^{d}, \forall (i,j) \in \mathcal{E}$, due to unique rigid formation with reference nodes using the Lemma \ref{lem1}, and the desired unique formation $\mathcal{F}(\boldsymbol{g}(t))$ is achieved. Bearing rigidity (as in Definition \ref{def_connectivity}) is preserved using Lemma \ref{lem2} by maintaining \cl{non-collinear reference} links $(i, j_1), (i, j_2)$. 
\end{proof}
By Fig. \ref{probset} we can maintain the \cl{non-collinear} reference links  using rigidity-critical target $\boldsymbol{p}_{iT}^r$ (\cl{the unique intersection of the two non-parallel lines $L_{j_1}^r$ and $L_{j_2}^r$} is discussed in Section \ref{FTSC: subsec/Connectivity Maintenance Position Det}). Also to achieve the desired bearing vectors $\boldsymbol{g}_{ij_1}^{d}$, $\boldsymbol{g}_{ij_2}^{d}$, the follower agent must reach the bearing-guided target $\boldsymbol{p}_{iT}^b$.
\section{CBF-inspired Rigidity Maintenance and Formation Tracking}
\label{FTSC: sec/rigiditywithf}
This section next develops the CBF-QP framework for Subproblem \ref{sp}, using formation and rigidity targets. 
\subsection{Determination of Bearing-Guided Target \cite[Sec. 4.1]{saharsh2026cbfformation}}
\label{FTSC: subsec/Bearing guided Position Det} 
Given the desired bearing vectors \( \boldsymbol{g}_{ij_k}^d = [g_{ij_k}^{d_x},\ g_{ij_k}^{d_y}]^{\top}, k \in \{1, 2\} ,\) their slope is defined as \( m_{ij_k}^d = \frac{g_{ij_k}^{d_y}}{g_{ij_k}^{d_x}} \).
The bearing-guided target is defined as follows.
\begin{definition}[Bearing-Guided Target \(\boldsymbol{p}_{iT}^b\)]
Given the desired formation position vector $\boldsymbol{p}_i^d$, for agent $i$, as defined in \eqref{dpos}, the bearing-guided target \(\boldsymbol{p}_{iT}^b\) for follower agent $i$ is defined as $\boldsymbol{p}_{iT}^b = \boldsymbol{p}_{i}^d - \boldsymbol{p}_i$, and explicitly computed as 
$
    \boldsymbol{p}_{iT}^b = A_i^{-1} \boldsymbol{\sigma}_i,
$
where $A_i = \begin{bmatrix}
        -m_{ij_1}^d & 1 \\
        -m_{ij_2}^d & 1
    \end{bmatrix}, \boldsymbol{\sigma}_i = 
        [[-m_{ij_1}^d,\ 1] \boldsymbol{p}_{ij_1},\ [-m_{ij_2}^d,\ 1] \boldsymbol{p}_{ij_2}]^{\top} .$
\end{definition}
Furthermore, the bearing-guided target $\boldsymbol{p}_{iT}^b$ can be interpreted as the unique point of intersection of the two non-parallel lines \( L_{j_1} \) and \( L_{j_2} \), where each line \( L_{j_k} \) passes through the point \(\boldsymbol{p}_{ij_k}\) with slope $m_{ij_k}^d$ (cf. Fig. \ref{probset}). 

Additionally, at \(\boldsymbol{p}_{iT}^b = \boldsymbol{0}\), i.e., $\boldsymbol{p}_i = \boldsymbol{p}_i^d$, the cosine distance \(\left(1 - \langle \boldsymbol{g}_{ij_k}, \boldsymbol{g}^d_{ij_k} \rangle \right)\) is zero, implying alignment with the desired bearing directions.

\subsection{Determination of Rigidity-Critical Target}\label{FTSC: subsec/Connectivity Maintenance Position Det}
We preserve initial bearing rigidity by ensuring the critical reference node remains within the sensing range. A reference node is \textit{critical} if it is most likely to violate this constraint, i.e., $j_{\text{crit}} := \begin{cases}\arg\max_{j \in \{j_1,j_2\}} \|\boldsymbol{p}_{ij}(t)\|, &\text{if} \max_j \|\boldsymbol{p}_{ij}(t)\|> R_{\text{crit}},\\
\emptyset, &\text{otherwise},\end{cases}$ where $R_{\text{crit}}=\lambda-\delta_{\text{safe}}$, and $\delta_{\text{safe}}>0$ is a safety margin. 
\cl{Also, to prevent collinearity, breaking rigidity with the references, we require $\sin (\nu_i(t)) \geq \sin (\nu_{\text{safe}})$, where the angle $\nu_i(t) := |\sin^{-1}(\boldsymbol{g}_{ij_1}(t) \times \boldsymbol{g}_{ij_1}(t))|$ and the safety margin $\nu_{\text{safe}} \in (0, \pi)$ which ensures that $\nu_i(t) \neq 0 \ \text{or} \ \pi$}.
Consider $\psi_{j_1j_2} := \tan^{-1}(\frac{y_{i2} - y_{i1}}{x_{i2} - x_{i1}})$.  We then show that each follower $i \in \mathcal{V}_F$ preserves rigidity by maintaining prescribed bearing relations with its references.
\begin{definition}[Rigidity-Critical Target $\boldsymbol{p}_{iT}^r$]
\label{rc_target}
    For any agent $i \in \mathcal{V}_F$ connected to a set of reference nodes $j_1, j_2$, $d := \|\boldsymbol{p}_{j_2j_1}\| = \|\boldsymbol{p}_{ij_1} - \boldsymbol{p}_{ij_2}\|$, with sensing parameter $\eta \in \mathbb{R}^{+}$ satisfying $\frac{d}{2} < \eta < R_{\text{crit}}$, $m_{ij_2}^r = \tan(\psi_{j_1j_2}-\phi_i), m_{ij_1}^r = \tan(\psi_{j_1j_2}+\phi_i), \phi_i = \cos^{-1}(\frac{d}{2\eta})$, the rigidity-critical target $\boldsymbol{p}_{iT}^r$ is defined as
$
    \boldsymbol{p}_{iT}^r = (A_i^r)^{-1} \boldsymbol{\sigma}_i^r,
$
where $A_i^r = \begin{bmatrix}
        -m_{ij_1}^r & 1 \\
        -m_{ij_2}^r & 1
    \end{bmatrix}, \boldsymbol{\sigma}_i^r = 
        [[-m_{ij_1}^r,\ 1] \boldsymbol{p}_{ij_1},\ [-m_{ij_2}^r,\ 1] \boldsymbol{p}_{ij_2}]^{\top} .$
\end{definition}
The rigidity-critical target $\boldsymbol{p}_{iT}^r$ can be interpreted as the apex of an isosceles triangle formed on the base segment $L_{j_1 j_2}$, joining the reference nodes $j_1$ and $j_2$, with equal side lengths $\eta$ (cf. Fig.~\ref{probset}). \cl{The positive perpendicular distance of the target from the base $L_{j_1 j_2}$, $h^r(d) := \sqrt{\eta^2 - (d/2)^2} > 0$, ensures non-collinearity with the references and we have hysteresis angle margin as the target apex angle $\sin (\phi^{r}) = \sin (\pi - 2\cos^{-1}(\frac{d}{2\eta})) > \sin (\nu_{\text{safe}})$.} Equivalently, $\boldsymbol{p}_{iT}^r$ is the unique intersection of the two non-parallel lines $L_{j_1}^r$ and $L_{j_2}^r$ passing through $\boldsymbol{p}_{ij_1}$ and $\boldsymbol{p}_{ij_2}$ with slopes $m_{ij_1}^r$ and $m_{ij_2}^r$, respectively. 

\subsection{Controller Design}
We propose to solve the Subproblem \ref{sp} using $\boldsymbol{p}_{iT}$: 
\begin{align}\label{unified_target}
\boldsymbol{p}_{iT} := 
\begin{cases} 
\boldsymbol{p}_{iT}^r, \text{if } j_{\text{crit}} \not = \emptyset \ \text{or} \ \cl{\sin (\nu_i(t)) < \sin (\nu_{\text{safe}})}, \\
\boldsymbol{p}_{iT}^b, \text{if } j_{\text{crit}} = \emptyset \ \& \ \cl{\sin (\nu_i(t)) \geq \sin (\nu_{\text{safe}})},
\end{cases}
\end{align}
for each follower agent $i\in \mathcal{V}_{F}$. The objective is to satisfy the local subtask by aligning the velocity $\boldsymbol{v}_i$ with the prescribed orientation given by $\boldsymbol{p}_{iT}$. 
To accommodate the change in target's speed, we define the target velocity:
\begin{align}\label{velocity_target}
\boldsymbol{v}_{T} := 
\begin{cases} 
\boldsymbol{v}_{j_{\text{crit}}}, \text{if } j_{\text{crit}} \not = \emptyset \ \text{or} \ \cl{\sin (\nu_i(t)) < \sin (\nu_{\text{safe}})} , \\
\boldsymbol{v}_{l}, \text{if } j_{\text{crit}} = \emptyset \ \& \ \cl{\sin (\nu_i(t)) \geq \sin (\nu_{\text{safe}})} .
\end{cases}
\end{align}
As the reference nodes $j_1,j_2 \in \mathcal{N}_i$ move, the unified target results in a piecewise continuous reference trajectory $\boldsymbol{p}_{iT}$ with discontinuous jumps, based on the set $j_{\text{crit}}$, for the change in the target from $\boldsymbol{p}_{iT}^r$ to $\boldsymbol{p}_{iT}^b$ or vice versa.
 \subsubsection*{Construction of Reach CBF}
We define the time stamps at which discontinuous jumps in $\boldsymbol{p}_{iT}$ occur as the sequence $(t_k)_{k \in \mathbb{N}}$, where initial time $t_1 = 0$ and target at each time stamp after jump as $\boldsymbol{p}_{iT}^{(k)} := \boldsymbol{p}_{iT}(t_k^{+}), \boldsymbol{p}_{iT}^{(1)} = \boldsymbol{p}_{iT}(0)$. For the required alignment between $\boldsymbol{v}_i$ and $\boldsymbol{p}_{iT}$, based on the angle $\theta_i$ between $\boldsymbol{v}_i$ and $\boldsymbol{p}_{iT}$, we define the sequence of reach CBFs $(b_i^{(k)}(t))_{k \in \mathbb{N}}, \forall t \in \Omega_k$, where $\Omega_k : =[t_{k}, t_{k+1})$ as
\begin{align}\label{cbfr}
    b_i^{(k)}(\hat{\boldsymbol{v}}_i,\hat{\boldsymbol{p}}_{iT}^{(k)}) := \langle\hat{\boldsymbol{v}}_i,\hat{\boldsymbol{p}}_{iT}^{(k)}\rangle - \cos(\theta_{it_{k}}),k \in \mathbb{N},
\end{align}
where $\theta_{it_{k}}=\cos^{-1}(\langle\hat{\boldsymbol{v}}_i(t_{k}),\hat{\boldsymbol{p}}_{iT}^{(k)}(t_{k})\rangle)$, $\hat{\boldsymbol{v}}_i=\frac{{\boldsymbol{v}}_i}{\|{\boldsymbol{v}}_i\|}$, $\hat{\boldsymbol{p}}_{iT}^{(k)}=\frac{{\boldsymbol{p}}_{iT}^{(k)}}{\|{\boldsymbol{p}}_{iT}^{(k)}\|}$, and $\max_{ \hat{\boldsymbol{v}}_i, \hat{\boldsymbol{p}}_{iT} \in \mathcal{S}} \  b_i^{(k)}(\hat{\boldsymbol{v}}_i,\hat{\boldsymbol{p}}_{iT}^{(k)})=1 - \cos(\theta_{it_{k}})$. The compact set $X_i^{(k)}$ is defined as $X_i^{(k)} := \{(\hat{\boldsymbol{v}}_i,\hat{\boldsymbol{p}}_{iT}^{(k)}) \in \mathcal{S} \times \mathcal{S}: b_i^{(k)}(\hat{\boldsymbol{v}}_i,\hat{\boldsymbol{p}}_{iT}^{(k)}) \geq 0\},$
and the boundary set $\partial X_i^{(k)} := \{(\hat{\boldsymbol{v}}_i,\hat{\boldsymbol{p}}_{iT}^{(k)}) \in \mathcal{S} \times \mathcal{S}: b_i^{(k)}(\hat{\boldsymbol{v}}_i,\hat{\boldsymbol{p}}_{iT}^{(k)}) = 0\}$. Here for given $\hat{\boldsymbol{p}}_{iT}^{(k)}$, $\forall \hat{\boldsymbol{v}}_i \in \mathcal{S}$, $b_i^{(k)}(\hat{\boldsymbol{v}}_i, \hat{\boldsymbol{p}}_{iT}^{(k)})$ has unique maxima as required by definition of the reach CBF \cite[Definition 3]{saharsh2026cbfformation}. Thus, we enforce the reachability of the target set $R_i^{(k)}:= \{(\hat{\boldsymbol{v}}_i,\hat{\boldsymbol{p}}_{iT}^{(k)}) \in \mathcal{S} \times \mathcal{S}: b_i^{(k)}(\hat{\boldsymbol{v}}_i,\hat{\boldsymbol{p}}_{iT}^{(k)}) \geq \mu_i^{(k)}\},$
where $\mu_i^{(k)} = \cos(\gamma_i^{(k)}) - \cos(\theta_{it_{k}}) > 0$, for the required alignment, with an arbitrarily chosen small angle $0<\gamma_i^{(k)} \leq \|\theta_{it_{k}}\|$, which results in $R_i^{(k)} \subseteq X_i^{(k)}$. Here, the parameter $\gamma_i^{(k)}$ represents the alignment error tolerance between $\hat{\boldsymbol{v}}_i$ and $\hat{\boldsymbol{p}}_{iT}^{(k)}$. 
Now, after simplifying the reach CBF constraint in \cite[Definition 3]{saharsh2026cbfformation}, with the reach CBF for target alignment as defined in \eqref{cbfr}, together with the system's dynamics as in (\ref{sys}), we obtain the linear constraint in $\boldsymbol{u}_i$ as $K_f^{(k)}({\boldsymbol{v}}_i, {\boldsymbol{p}}_{iT}^{(k)}) + (K_h^{(k)})^{\top}({\boldsymbol{v}}_i, {\boldsymbol{p}}_{iT}^{(k)})\boldsymbol{u}_i \geq \frac{\mu_i^{(k)}}{\tau_i^{(k)}}, \forall t \in \Omega_k$ with
\begin{align}\label{movet}
&(K_h^{(k)})^{\top}({\boldsymbol{v}}_i, {\boldsymbol{p}}_{iT}^{(k)}) = (\hat{\boldsymbol{p}}_{iT}^{(k)})^{\top}\frac{1}{\|\boldsymbol{v}_i\|}\left[I_2 - \hat{\boldsymbol{v}}_i\hat{\boldsymbol{v}}_i^{\top}\right]h_{v_i}, \nonumber \\
&K_f^{(k)}({\boldsymbol{v}}_i, {\boldsymbol{p}}_{iT}^{(k)}) = (\hat{\boldsymbol{p}}_{iT}^{(k)})^{\top}\frac{1}{\|\boldsymbol{v}_i\|}\left[I_2 - \hat{\boldsymbol{v}}_i\hat{\boldsymbol{v}}_i^{\top}\right]f_{v_i} + \hat{\boldsymbol{v}}_i^{\top}{\dot{\hat{\boldsymbol{p}}}_{iT}^{(k)}}.
\end{align}
Using CBF-QP in \cite[Theorem 1]{saharsh2026cbfformation} and the constraint \eqref{movet}, our goal is to reach the set $R_i^{(k)}$ in time $\tau_i^{(k)}$, where $\hat{\boldsymbol{v}}_i$ and $\hat{\boldsymbol{p}}_{iT}^{(k)}$ are aligned with small angle error $\gamma_i^{(k)}$, i.e., $\cos(\theta_i(\tau_i^{(k)})) \geq \cos(\gamma_i^{(k)})$, while being within the compact set $X_i^{(k)}$. 
\begin{figure*}[t]
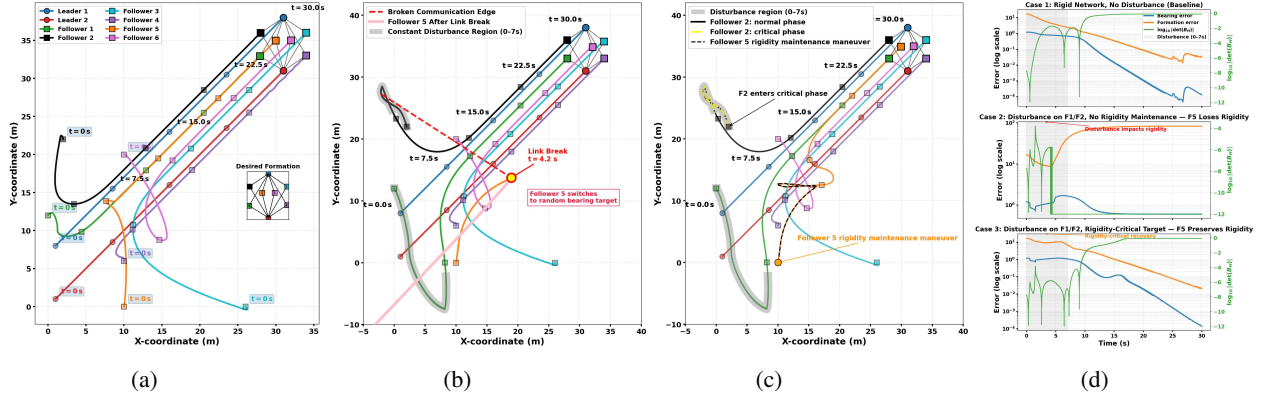

\begin{overpic}[width=\linewidth]{fig2.pdf}
    \put(10, -3){\small (a)}   
    \put(35, -3){\small (b)}
    \put(60, -3){\small (c)}
    \put(86, -3){\small (d)}
\end{overpic}
\vspace{2pt}
    \caption{Simulation validating our approach for formation tracking with rigidity maintenance under disturbance with 8 agents. Video: (a) https://bit.ly/43JrDUl, (b) https://bit.ly/4vmm6zg, (c) https://bit.ly/3QPdCS7}
    \label{fig2}
\vspace{-5mm}
\end{figure*}

\subsubsection*{Reach CBF-based Target Tracking}
We now formulate a QP optimization problem for the agent $i$, using constraint \eqref{movet} stitched for all intervals $\Omega_k, \forall k \in \mathbb{N}$.
\begin{assum} \label{FTSC: assum/anti-align}
    The bearing-guided target direction \(\hat{\boldsymbol{p}}_{iT}\) and the agent’s velocity \(\boldsymbol{v}_i\) follows $\langle\hat{\boldsymbol{p}}_{iT}^{(k)}, \hat{\boldsymbol{v}}_i(t)\rangle \neq -1, \forall t \in \mathbb{R}_{0}^{+}, \forall k \in \mathbb{N}$.
\end{assum}
The constraint \eqref{movet} becomes invalid in the degenerate case when $\left[I_2 - \boldsymbol{v}_i\boldsymbol{v}_i^{\top}\right]h_{v_i} = 0_{2}$ or $\langle\boldsymbol{p}_{iT}^{(k)},\boldsymbol{v}_i^{\perp}\rangle = 0$, where the vector $\boldsymbol{v}_i^{\perp}$ is the orthogonal component of $\boldsymbol{v}_i$. However, $\langle\boldsymbol{p}_{iT}^{(k)},\boldsymbol{v}_i^{\perp}\rangle = 0$ implies either alignment, which requires no correction, or anti-alignment, which is excluded by Assumption \ref{FTSC: assum/anti-align}. Since the system (\ref{sys}) is controllable ($Null(h_{v_i}) = \emptyset$), the constraint remains valid (i.e. $\left[I_2 - \boldsymbol{v}_i\boldsymbol{v}_i^{\top}\right]h_{v_i} \neq 0_{2}$). The QP-based control formulation is presented below. 
\begin{theorem}
    Consider the heterogeneous MAS (modeled as \eqref{sys}) connected under a leader-follower framework $\mathcal{G}(t)$, satisfying Assumptions \ref{A1}-\ref{FTSC: assum/anti-align} with limited sensing range $\lambda$. Then for any $i \in \mathcal{V}_F$, given the reach CBF $b_i^{(k)}(\hat{\boldsymbol{v}}_i,\hat{\boldsymbol{p}}_{iT}^{(k)})$ as in \eqref{cbfr}, along with the constraint \eqref{movet}, using $\boldsymbol{p}_{iT}(t), t \in \Omega_k, \forall k \in \mathbb{N}$ as in \eqref{unified_target}, if each follower agent applies the control input $\boldsymbol{u}_i$, $\forall t \in \Omega_k, \forall k \in \mathbb{N}$:  
\begin{align}\label{QPws}
    &{\boldsymbol{u}}_i = \arg \min_{\boldsymbol{q} \in \mathbb{R}^2} \quad \frac{1}{2}(\boldsymbol{q} -  \boldsymbol{u}_{iv})^{\top}(\boldsymbol{q}-  \boldsymbol{u}_{iv}), \nonumber \\
    &\text{s.t.} \quad K_{f}^{(k)}({\boldsymbol{v}}_i, {\boldsymbol{p}}_{iT}^{(k)}) + (K_{h}^{(k)})^{\top}({\boldsymbol{v}}_i, {\boldsymbol{p}}_{iT}^{(k)})\boldsymbol{q} \geq \frac{\mu_i^{(k)}}{\tau_i^{(k)}}, 
\end{align}
with $K_{f}^{(k)}, K_{h}^{(k)}$ as in \eqref{movet}, along with the controller $\boldsymbol{u}_{iv} = h_{v_i}^{-1}(-f_{v_i}  - K_{i1}\boldsymbol{\xi}_{ip} + K_{i1}(\boldsymbol{v}_T - \boldsymbol{v}_i) + K_{i2}\boldsymbol{\xi}_{iv}),$
where $K_{i2} = \kappa_{i2}I_{2},K_{i1} = \kappa_{i1}I_{2}, 
\kappa_{i1}>1,\kappa_{i2}>1+\frac{\kappa_{i1}}{2}, \boldsymbol{\xi}_{ip} = \boldsymbol{p}_{iT}, \boldsymbol{\xi}_{iv} = \boldsymbol{v}_T - \boldsymbol{v}_i + K_{i1}\boldsymbol{\xi}_{ip}$, \cl{and
$\zeta_i\sqrt{ \Delta_j^2/2\alpha} < h^r(d) - \sin\nu_{\text{safe}}\, R_{\text{crit}}^2/d, V_i^{(k)}(t_k^+)\le \zeta_iV_i^{(k)}(t_k^-), \forall k \in \mathbb{N}, \text{where} \  V_i^{(k)} = \frac{1}{2}\boldsymbol{\xi}_{ip}^{\top}\boldsymbol{\xi}_{ip} + \frac{1}{2}\boldsymbol{\xi}_{iv}^{\top}\boldsymbol{\xi}_{iv}, \zeta_i \geq 1, \Delta_j = (\Gamma_{jf}v_{max} + (\Gamma_{jh}v_{max} 
 + 1)\|\boldsymbol{u}_j\|), \alpha = \min \{-\frac{(1 - \kappa_{i1})}{2}, -\frac{(2 + \kappa_{i1}- 2\kappa_{i2})}{2}\}$}, then the underlying communication graph preserves rigidity as in Definition \ref{def_connectivity} and unique formation is achieved under bearing rigid graph as in Definition \ref{mform}. Thus, Problem \ref{mainprob} is solved.
\label{FTSC: thm/QP}
\end{theorem}
\begin{proof}
With initial $\boldsymbol{v}_i(t_{k})$, the barrier function $b_i^{(k)}(\boldsymbol{v}_i(t_{k}),\boldsymbol{p}_{iT}(t_{k})) = 0, \forall k \in \mathbb{N}$, and the set $R_i^{(k)}$ is reached by integrating the constraint in \eqref{movet} over $[t_{k},\tau_i^{(k)}] \subseteq \Omega_k$ with lower bound $\frac{\mu_i^{(k)}}{\tau_i^{(k)}}$, and using the following arguments analogous to the proof of \cite[Theorem 1]{saharsh2026cbfformation}, it follows that $b_i^{(k)}(\hat{\boldsymbol{v}}_i(\tau_i^{(k)})) \geq \mu_i^{(k)}$, ensuring $\cos(\theta_i(\tau_i^{(k)})) \geq \cos(\gamma_i^{(k)})$. Additionally, we need the relative target position $\boldsymbol{\xi}_{ip} \rightarrow \boldsymbol{0}$ and the relative target velocity $\boldsymbol{\xi}_{iv} \rightarrow \boldsymbol{0}$ as $ (\boldsymbol{\xi}_{ip},\boldsymbol{\xi}_{iv}) = (0,0) \Rightarrow \boldsymbol{p}_{iT} = \boldsymbol{0}, \boldsymbol{v}_{T} = \boldsymbol{v}_i \Rightarrow \boldsymbol{g}_{ij} = \boldsymbol{g}_{ij}^d,\forall j \in \{j_1,j_2\}$. 
 Consider a Lyapunov function $V_i^{(k)}: \mathbb{R}^2 \times \mathbb{R}^2 \rightarrow \mathbb{R}_0^+$, $ V_i^{(k)} = \frac{1}{2}\boldsymbol{\xi}_{ip}^{\top}\boldsymbol{\xi}_{ip} + \frac{1}{2}\boldsymbol{\xi}_{iv}^{\top}\boldsymbol{\xi}_{iv} , \forall t \in \Omega_k $.
 With the system dynamics \eqref{sys}, $f_{v_i}(\boldsymbol{0}) = \boldsymbol{0}, h_{v_i}(\boldsymbol{0}) = I_2, \forall i$, we consider the nominal error dynamics written as $\dot{\boldsymbol{\xi}}_{ip} = \boldsymbol{v}_T - \boldsymbol{v}_i = -K_{i1}\boldsymbol{\xi}_{ip} + \boldsymbol{\xi}_{iv}$, and if $j_{\text{crit}} = \emptyset$ and \cl{$(\sin (\nu_i(t)) \geq \sin (\nu_{\text{safe}}))$}, $\dot{\boldsymbol{\xi}}_{iv} = K_{i1}\dot{\boldsymbol{\xi}}_{ip} -f_{v_i} - h_{v_i}\boldsymbol{u}_i = K_{i1}\boldsymbol{\xi}_{ip} - K_{i2}\boldsymbol{\xi}_{iv}$ and otherwise, $\|\dot{\boldsymbol{\xi}}_{iv}\| = \|K_{i1}\dot{\boldsymbol{\xi}}_{ip} + f_{v_j}+h_{v_j}\boldsymbol{u}_j-f_{v_i} - h_{v_i}\boldsymbol{u}_i\| \leq \|K_{i1}\boldsymbol{\xi}_{ip} - K_{i2}\boldsymbol{\xi}_{iv}\| + \Delta_j, \Delta_j = (\Gamma_{jf}v_{max} + (\Gamma_{jh}v_{max} 
 + 1)\|\boldsymbol{u}_j\|)$. By differentiating Lyapunov and substituting the error dynamics, we get
$\dot{V}^{(k)} = \boldsymbol{\xi}_{ip}^{\top}\dot{\boldsymbol{\xi}}_{ip} + \boldsymbol{\xi}_{iv}^{\top}\dot{\boldsymbol{\xi}}_{iv}  \leq \boldsymbol{\xi}_{ip}^{\top}\boldsymbol{\xi}_{iv} - \boldsymbol{\xi}_{ip}^{\top}K_{i1}\boldsymbol{\xi}_{ip} + \boldsymbol{\xi}_{iv}^{\top}K_{i1}\boldsymbol{\xi}_{ip} - \boldsymbol{\xi}_{iv}^{\top}K_{i2}\boldsymbol{\xi}_{iv} + \Delta_j\|\boldsymbol{\xi}_{iv}\|\leq  \frac{(1 + \kappa_{i1})}{2}(\|\boldsymbol{\xi}_{ip}\|^{2} + \|\boldsymbol{\xi}_{iv}\|^{2}) - \kappa_{i1}\|\boldsymbol{\xi}_{ip}\|^{2} - \kappa_{i2}\|\boldsymbol{\xi}_{iv}\|^{2} + \frac{\Delta_j^2}{2} + \frac{\|\boldsymbol{\xi}_{iv}\|^{2}}{2}\leq \frac{(1 - \kappa_{i1})}{2}\|\boldsymbol{\xi}_{ip}\|^{2} + \frac{(2 + \kappa_{i1}- 2\kappa_{i2})}{2}\|\boldsymbol{\xi}_{iv}\|^{2} + \frac{\Delta_j^2}{2}$,
where the upper bound was obtained using Young's inequality and the control gain $\kappa_{i1}>0$.
Since $\kappa_{i1}>1,\kappa_{i2}>1+\frac{\kappa_{i1}}{2}$, we have $\dot{V}_i^{(k)} < -2\alpha{V}_i^{(k)} + \frac{\Delta_j^2}{2}, \alpha = \min \{-\frac{(1 - \kappa_{i1})}{2}, -\frac{(2 + \kappa_{i1}- 2\kappa_{i2})}{2}\}$. So, the error $\|[\boldsymbol{\xi}_{ip}^{\top},\boldsymbol{\xi}_{iv}^{\top}]\|$ is uniformly ultimately bounded with $\|[\boldsymbol{\xi}_{ip}^{\top},\boldsymbol{\xi}_{iv}^{\top}]\| \leq \sqrt\frac{\Delta_j^2}{2\alpha}$ and in case $j_{\text{crit}} = \emptyset, \Delta_j = 0$, the equilibrium $(\boldsymbol{\xi}_{ip}^{\top},\boldsymbol{\xi}_{iv}) = (\boldsymbol{0}, \boldsymbol{0})$ is exponentially stable.\\
Since $\boldsymbol{p}_{iT}$ switches with the presence of a critical reference node, the closed-loop dynamics constitute a switched nonlinear system. 
We model the target trajectory as a switched non-linear system. The analysis follows as a direct consequence of the well-known switched-systems results in Section 3.2 of Switching in Systems and Control by Daniel Liberzon. Let the switching signal $\sigma_i(t)\in\{r,b\}$ denote the active subsystem for follower agent $i$, where
\[
\sigma_i(t)=
\begin{cases}
r, & \text{if } j_{\text{crit}} \not = \emptyset \ \text{or} \ {(\sin (\nu_i(t)) < \sin (\nu_{\text{safe}}))},\\
b, & \text{if } j_{\text{crit}} = \emptyset \ \text{and} \ {(\sin (\nu_i(t)) \geq \sin (\nu_{\text{safe}}))}.
\end{cases}
\]
The corresponding target definitions are $\boldsymbol{p}_{iT}^{r}, \boldsymbol{p}_{iT}^{b},$ with target velocities $\boldsymbol{v}_{T}^{r}=\boldsymbol{v}_{j_{\text{crit}}}, \boldsymbol{v}_{T}^{b}=\boldsymbol{v}_l.$ Define the subsystem-dependent tracking errors as
\begin{align}
\boldsymbol{\xi}_{ip}^{\sigma_i} &= \boldsymbol{p}_{iT}^{\sigma_i},\\
\boldsymbol{\xi}_{iv}^{\sigma_i} &= \boldsymbol{v}_{T}^{\sigma_i}-\boldsymbol{v}_i+K_{i1}\boldsymbol{\xi}_{ip}^{\sigma_i},
\end{align}
with the Lyapunov function candidate
\begin{align}\label{switch_lya}
V_i^{\sigma_i}
=
\frac{1}{2}\|\boldsymbol{\xi}_{ip}^{\sigma_i}\|^2
+
\frac{1}{2}\|\boldsymbol{\xi}_{iv}^{\sigma_i}\|^2.
\end{align}

Using the nominal controller (17), the closed-loop error dynamics for both subsystems can be written as
\begin{align}
\dot{\boldsymbol{\xi}}_{ip}^{\sigma_i}
&=
-K_{i1}\boldsymbol{\xi}_{ip}^{\sigma_i}
+
\boldsymbol{\xi}_{iv}^{\sigma_i},
\\
\dot{\boldsymbol{\xi}}_{iv}^{\sigma_i}
&=
K_{i1}\boldsymbol{\xi}_{ip}^{\sigma_i}
-
(K_{i2}\boldsymbol{\xi}_{iv}^{\sigma_i} + P^{\sigma_i}),
\end{align}
where $P^{r} = 0, P^{b} = f_{v_j} + h_{v_j}\boldsymbol{u}_j$
Differentiating \eqref{switch_lya} along the trajectories yields
\begin{align}
\dot V_i^{\sigma_i}
&=
(\boldsymbol{\xi}_{ip}^{\sigma_i})^\top
\dot{\boldsymbol{\xi}}_{ip}^{\sigma_i}
+
(\boldsymbol{\xi}_{iv}^{\sigma_i})^\top
\dot{\boldsymbol{\xi}}_{iv}^{\sigma_i}
\nonumber\\
&\leq
-(\boldsymbol{\xi}_{ip}^{\sigma_i})^\top
K_{i1}\boldsymbol{\xi}_{ip}^{\sigma_i}
+
(\boldsymbol{\xi}_{ip}^{\sigma_i})^\top
\boldsymbol{\xi}_{iv}^{\sigma_i}
+
(\boldsymbol{\xi}_{iv}^{\sigma_i})^\top
K_{i1}\boldsymbol{\xi}_{ip}^{\sigma_i}
-
(\boldsymbol{\xi}_{iv}^{\sigma_i})^\top
K_{i2}\boldsymbol{\xi}_{iv}^{\sigma_i} + \Delta_j\|\boldsymbol{\xi}_{iv}\|,
\end{align}
where $\Delta_j = (\Gamma_{jf}v_{max} + (\Gamma_{jh}v_{max} 
 + 1)\|\boldsymbol{u}_j\|)$.
Using Young's inequality, together with $K_{i1}=\kappa_{i1}I_2, K_{i2}=\kappa_{i2}I_2,$ we obtain
\begin{align}
\dot V_i^{\sigma_i}
\leq
-\frac{\kappa_{i1}-1}{2}
\|\boldsymbol{\xi}_{ip}^{\sigma_i}\|^2
-
\frac{2\kappa_{i2}-(1+\kappa_{i1})}{2}
\|\boldsymbol{\xi}_{iv}^{\sigma_i}\|^2+ \frac{\Delta_j^2}{2} + \frac{\|\boldsymbol{\xi}_{iv}\|^{2}}{2}.
\end{align}
Define $\alpha = \min \{-\frac{(1 - \kappa_{i1})}{2}, -\frac{(2 + \kappa_{i1}- 2\kappa_{i2})}{2}\}.$ Then
\begin{align}\label{exp_decay}
\dot V_i^{\sigma_i}
\leq
-2\alpha_iV_i^{\sigma_i} + \frac{\Delta_j^2}{2},
\end{align}
provided $\kappa_{i1}>1, \kappa_{i2}>1 + \frac{\kappa_{i1}}{2}.$ Thus, each subsystem is uniformly ultimately bounded. Next, we analyze the effect of switching. Let $t_k$ denote a switching instant. Since the target changes discontinuously from $\boldsymbol{p}_{iT}^{r}$ to $\boldsymbol{p}_{iT}^{b}$ (or vice versa), the error coordinates undergo a bounded jump. Define the augmented error state $x_i^{\sigma_i}
=
\begin{bmatrix}
\boldsymbol{\xi}_{ip}^{\sigma_i}\\
\boldsymbol{\xi}_{iv}^{\sigma_i}
\end{bmatrix}.$
At switching time $t_k$, $x_i^{\sigma_i(t_k^+)}
=
x_i^{\sigma_i(t_k^-)}
+
\Delta_i,$ where the switching offset satisfies
\begin{align}
\Delta_i
=
\begin{bmatrix}
\boldsymbol{p}_{iT}^{b}-\boldsymbol{p}_{iT}^{r}\\
(\boldsymbol{v}_{l}-\boldsymbol{v}_{j_{\text{crit}}})
+
K_{i1}(\boldsymbol{p}_{iT}^{b}-\boldsymbol{p}_{iT}^{r})
\end{bmatrix},
\end{align}
with finite norm $D_i:=\|\Delta_i\|<\infty.$
The boundedness assumption $D_i=\|\Delta_i\|<\infty$ is naturally satisfied in the considered multi-agent setting. Since all agents operate within a confined workspace, their positions remain uniformly bounded for all time. Furthermore, the agent velocities are constrained by a known maximum speed $v_{\max}$ due to actuator limits. Consequently, both the target mismatch term $(\boldsymbol{p}_{iT}^{b}-\boldsymbol{p}_{iT}^{r})$ and the velocity difference term $(\boldsymbol{v}_l-\boldsymbol{v}_{j_{\text{crit}}})$ remain bounded for all switching instants $t_k$. In addition, the gain matrix $K_{i1}$ is constant and finite. Therefore, the switching-induced state jump $\Delta_i$ is uniformly bounded, ensuring the existence of a finite constant $D_i$ independent of time and switching index $k$.
Using the definition of the Lyapunov functions,
\begin{align}
V_i^{\sigma_i(t_k^+)}
&=
\frac{1}{2}
\|x_i^{\sigma_i(t_k^-)}+\Delta_i\|^2
\nonumber\\
&=
V_i^{\sigma_i(t_k^-)}
+
(x_i^{\sigma_i(t_k^-)})^\top\Delta_i
+
\frac{1}{2}D_i^2.
\end{align}

Applying the Cauchy--Schwarz inequality, $(x_i^{\sigma_i(t_k^-)})^\top\Delta_i
\leq
\|x_i^{\sigma_i(t_k^-)}\|D_i,$ and using $\|x_i^{\sigma_i(t_k^-)}\|
=
\sqrt{2V_i^{\sigma_i(t_k^-)}},$ we obtain $V_i^{\sigma_i(t_k^+)}
\leq
V_i^{\sigma_i(t_k^-)}
+
D_i\sqrt{2V_i^{\sigma_i(t_k^-)}}
+
\frac{1}{2}D_i^2.$ To establish stability of the switched closed-loop system, we now derive a lower bound on the dwell time between consecutive switching instants.

From Theorem 3, once agent $i$ has the required reference neighbors, the inter-agent distances remain strictly inside the sensing region with a positive sensing and angle margin. For instance if the sensing margin for agent $i$ at $t_k$ from mode $r \rightarrow b$ (due to violation of sensing condition) as $\delta_i
:=
\min_{j \in \{j_1, j_2\}}
\left(R_{\text{crit}}-\|\boldsymbol p_i(t_k)-\boldsymbol p_j(t_k)\|
\right)$.
Assume the agent velocities are uniformly bounded as $\|\boldsymbol v_i(t)\|
\leq
v_{\max},
\qquad
\forall i\in\mathcal V.$ Then the maximum rate at which an agent can move toward the sensing boundary is bounded by $v_{\max}$. Therefore, to lose or gain the criticality of the neighbor connection, the inter-agent distance must increase or decrease by at least $\delta_i$, requiring a minimum time duration satisfying
\begin{align}
\tau_{d,i}
\geq
\frac{\delta_i}{2v_{\max}}.
\label{mindwellbound}
\end{align}
Thus, arbitrarily fast switching is excluded and the switching signal admits a strictly positive dwell time.

Next, integrating \eqref{exp_decay} over the interval $\Omega_k=[t_k,t_{k+1})$ yields
\begin{align}
V_i^{\sigma_i}(t_{k+1}^{-})
\leq
e^{-2\alpha_i(t_{k+1}-t_k)}
V_i^{\sigma_i}(t_k^{+}) + \sqrt\frac{\Delta_j^2}{2\alpha}.
\end{align}
Using the minimum dwell-time bound \eqref{mindwellbound},
\begin{align}
V_i^{\sigma_i}(t_{k+1}^{-})
\leq
e^{-2\alpha_i\tau_{d,i}}
V_i^{\sigma_i}(t_k^{+}) + \sqrt\frac{\Delta_j^2}{2\alpha}
\leq
e^{-2\alpha_i\frac{\delta_i}{2v_{\max}}}
V_i^{\sigma_i}(t_k^{+}) + \sqrt\frac{\Delta_j^2}{2\alpha}.
\label{dwell_decay_final}
\end{align}

Equation \eqref{dwell_decay_final} shows that the Lyapunov function exponentially decreases during each dwell-time interval before another switch can occur. Since the switching-induced jump satisfies
\begin{align}
V_i^{\sigma_i(t_k^+)}
-
V_i^{\sigma_i(t_k^-)}
\leq
D_i\sqrt{2V_i^{\sigma_i}(t_k^-)}
+
\frac{1}{2}D_i^2,
\end{align}
the switching transient remains bounded for finite $D_i$.

To complete the stability argument, we now combine the inter-switch ultimate bound and the switching-induced jump to obtain a recursive bound on the Lyapunov function at switching instants.

From the previous result, define the post-jump value at switching instant $t_k$ as
\begin{align}
V_i^{\sigma_i}(t_k^+)
\leq
V_i^{\sigma_i}(t_k^-)
+
D_i\sqrt{2V_i^{\sigma_i}(t_k^-)}
+
\frac{1}{2}D_i^2.
\end{align}

For any $V_i^{\sigma_i}(t_k^-)\geq V_{\min}>0$, apply Young's inequality to obtain
\begin{align}
D_i\sqrt{2V_i^{\sigma_i}(t_k^-)}
\leq
\epsilon V_i^{\sigma_i}(t_k^-)
+
\frac{D_i^2}{2\epsilon},
\qquad \epsilon>0.
\end{align}

Thus,
\begin{align}
V_i^{\sigma_i}(t_k^+)
\leq
(1+\epsilon)V_i^{\sigma_i}(t_k^-)
+
\frac{D_i^2}{2}\left(1+\frac{1}{\epsilon}\right).
\end{align}

Using $V_i^{\sigma_i}(t_k^-)\geq V_{\min}$, we normalize the additive term as
\begin{align}
\frac{D_i^2}{2}\left(1+\frac{1}{\epsilon}\right)
\leq
\frac{D_i^2}{2V_{\min}}\left(1+\frac{1}{\epsilon}\right)V_i^{\sigma_i}(t_k^-).
\end{align}

Hence, the switching jump admits the multiplicative bound $V_i^{\sigma_i}(t_k^+)
\leq
\zeta_i V_i^{\sigma_i}(t_k^-),$ where\\
$\underbrace{\zeta_i
:=
(1+\epsilon)
+
\frac{D_i^2}{2V_{\min}}\left(1+\frac{1}{\epsilon}\right).}_{computed}$ Combining this with the inter-switch ultimate bound \eqref{dwell_decay_final}, we obtain
\begin{align}
V_i^{\sigma_i}(t_{k+1}^-)
\leq
\zeta_i
e^{-2\alpha_i\frac{\delta_i}{2v_{\max}}}
V_i^{\sigma_i}(t_k^-) + \zeta_i \sqrt\frac{\Delta_j^2}{2\alpha}.
\end{align}

Define the contraction factor $r_i
:=
\zeta_i
e^{-2\alpha_i\frac{\delta_i}{2v_{\max}}}.$ For ultimate bound under switching, it is sufficient that $r < 1.$ This yields the explicit gain and dwell-time condition $2\alpha_i\frac{\delta_i}{2v_{\max}}
>
\ln(\zeta_i).$ Substituting the definition of $\alpha_i$, the sufficient conditions on the control gains become $\kappa_{i1}
>
1
+
\frac{v_{\max}}{\delta_i}
\ln(\zeta_i),$ and $\kappa_{i2}
>
1 + \frac{\kappa_{i1}}{2}
+
\frac{v_{\max}}{2\delta_i}
\ln(\zeta_i).$ Therefore, we obtain the dwell-time bound induced by the sensing margin, angular margin, and the gain conditions above.
Therefore, the switching signal satisfies a practical dwell-time condition and a bounded switching-based Lyapunov jump.  
Moreover, $\max\{\|\boldsymbol{p}_{ij_1}(t)\|,\|\boldsymbol{p}_{ij_2}(t)\|\}<\eta<R_{\text{crit}}$ ensures a positive sensing margin $R_{\text{crit}}-\eta$, and \cl{$\sin(\phi^r) > \sin(\nu_{\text{safe}})$ ensures a positive hysteresis margin}, yielding an implicit dwell time. Although switching causes discontinuous changes in the target definition, the corresponding Lyapunov jump remains bounded as $V_i(t_k^+)\le \zeta_iV_i(t_k^-)$ for some finite $\zeta_i\ge1$ (as agents operate in confined workspace with limited speed), while within each dwell interval the Lyapunov function decays asymptotically as $\dot V_i < 0$, ensuring that the decay dominates the bounded switching transient (as in \cite{liberzon2003switching}[Section 3.2]).\\
\cl{Furthermore, agent $i$ remains bounded away from collinearity with $j_1,j_2$ for all $t$. Using the
triangle-area identity, the angle $\nu_i(t)$ subtended by $j_1,j_2$ at agent $i$ satisfies $\sin\nu_i(t) = \frac{d(t) . h_i(t)}{\|\boldsymbol{p}_{ij_1}(t)\|\,\|\boldsymbol{p}_{ij_2}(t)\|}$, where $h_i(t)$ is the agent $i$'s perpendicular distance from $L_{j_1j_2}$. From the uniform
ultimate bound established in the proof of Theorem \ref{FTSC: thm/QP}, together with the bounded
Lyapunov jump $V_i(t_k^+)\le \zeta_iV_i(t_k^-)$ at switching instants, the tracking error satisfies
$\rho(t) := \|\boldsymbol{\xi}_{ip}(t)\| \le \zeta_i\sqrt{ \Delta_j^2/2\alpha}$, so 
$h_i(t) \geq h^r(d) - \rho(t)$, for the worst-case when $\rho(t)$ is aligned with the perpendicular to $L_{j_1j_2}$. Using
$\|\boldsymbol{p}_{ij_1}(t)\|,\|\boldsymbol{p}_{ij_2}(t)\| \le R_{\text{crit}}$, this yields $\sin\nu_i(t) \geq \frac{d\big(h^r(d) - \rho(t)\big)}{R_{\text{crit}}^2}.$
Thus, with the sufficient condition on the design parameters satisfying $\zeta_i\sqrt{ \Delta_j^2/2\alpha} < h^r(d) - \sin\nu_{\text{safe}}\, R_{\text{crit}}^2/d$, it
gives $\sin\nu_i(t) > \sin\nu_{\text{safe}} > 0$, so agent $i$ never crosses the angular
switching threshold nor becomes collinear with $j_1,j_2$. }

 Thus, the controller tracks $\boldsymbol{p}_{iT}^r$ to satisfy $\max\{\|\boldsymbol{p}_{ij_1}(t)\|,\ \|\boldsymbol{p}_{ij_2}(t)\|\} < R_{\text{crit}}$ in finite time by Definition \ref{rc_target} maintaining the reference links, \cl{preventing collinearity as $\sin (\nu_i(t)) \geq \sin (\nu_{\text{safe}})$}, and subsequently tracks $\boldsymbol{p}_{iT}^b$ to achieve the formation subtask $\mathcal{F}_i(\boldsymbol{g}(t))$. Consequently, by Theorem \ref{theo1}, the overall formation task $\mathcal{F}(\boldsymbol{g}(t))$ is achieved, solving Problem \ref{mainprob}.
\end{proof}

\section{Results \& Discussions}
\label{sec: results}
We validate our proposed approach, with simulation including two leader agents $\{1,2\}$ and six follower agents $\{3,4,5,6,7,8\}$ in a 2D space, connected as ${\mathcal{G}}(0) = (\{1,2,3,4,5,6,7,8\},\{(1,2),(3,1),(3,2),(4,1),(4,2),(5,1),\\(5,2),(6,1),(6,2),(7,3),(7,4),(8,5),(8,6)\})$, satisfying Assumption \ref{A1} and \cl{connected to their reference nodes (low-level neighbors) }. The formation task $\mathcal{F}(\boldsymbol{g}(t))$ is defined by the bearing configuration $\boldsymbol{g}^{d}$ to attain the desired shape shown in Fig. \ref{fig2}(a). The follower agents 3, 6, and 8 are modeled as a double integrator with $f_{v_i}=\boldsymbol{0}, h_{v_i}=I_2$ and agents 4, 5, and 7 are modeled as differential drive (or Ackermann drive with small slip angle), with $f_{v_i} = \boldsymbol{0}, h_{v_i} = \begin{bmatrix}
    \cos(\delta_i) & -\beta_i\sin(\delta_i)\\
    \sin(\delta_i) & \beta_i\cos(\delta_i)
\end{bmatrix}$, where $\delta_i$ is the robot's orientation, $\dot{\delta_i} = [0,\ 1]\boldsymbol{u}_i$,  $\beta_i = \|\boldsymbol{v}_i\|$ {(refer \cite{sawarkar2026sliding}[II B]).} For double integrator system, $\Gamma_{if} = 0, \Gamma_{ih} = 1$ and for differential drive system, $\Gamma_{if} = 0, \Gamma_{ih} = v_{max}$, where $v_{max}$ is the maximum velocity.
 Leader agents follow predetermined trajectories (shown by red and blue lines in Fig. \ref{fig2}), while follower agents aim to achieve a desired formation using Theorem \ref{FTSC: thm/QP}. We consider three case study: Case (i) \textit{Formation tracking without disturbance}: All agents successfully track the desired moving formation, as shown in Fig.~\ref{fig2}(a). Using Theorem~\ref{FTSC: thm/QP}, both bearing error and formation error converge to small values ( $ \approx 10^{-2}$), as illustrated in Fig.~\ref{fig2}(d). Case (ii) \textit{Formation tracking with disturbance}: Under a constant disturbance applied from 0--7 s, with perturbations $[-0.05, 0.04], [-0.01, 0.04]$ added to the positions of Follower 1, Follower 2 (Agent 3, 4), respectively, Follower 5 (Agent 7) loses connectivity with Follower 2, leading to a loss of network rigidity as $\det(\mathcal{B}_{ff}) \to 0$ (i.e., $\log_{10}(|\det(\mathcal{B}_{ff})|) \approx -12$, as shown in Fig.~\ref{fig2}(d)). Although the bearing error remains small, the mean formation error increases significantly due to the loss of rigidity, and Follower 7 deviates from the desired formation (Fig.~\ref{fig2}(b)). Case (iii) \textit{Formation tracking with disturbance and rigidity maintenance}: Under the same disturbance, a rigidity-critical link between Follower 4 and Follower 7 is preserved. As a result, the formation remains rigid and is successfully tracked, as shown in Fig.~\ref{fig2}(c). \\
\textit{Comparison with existing works}: Existing rigidity maintenance strategy such as \cite{zelazo2015decentralized} uses the rigidity eigenvalue $\mu_7 > 0$ of the symmetric rigidity matrix, requiring each agent to run a tenth-order dynamic estimator via power iteration and PI average consensus filters, with convergence requiring strict eigenvalue separation $|\mu_7 - \mu_8| > 0$ which cannot be guaranteed under time-varying topologies. Connectivity-based approaches~\cite{bhatia2025decentralized} maintain the Fiedler eigenvalue $\mu_2 > 0$, ensuring only graph connectivity, a strictly weaker condition than rigidity, since a connected graph need not satisfy $\det(\mathcal{B}_{ff}) \neq 0$. In contrast, the proposed framework avoids global spectral estimation entirely by identifying rigidity-critical reference nodes $j_1, j_2 \in \mathcal{N}_i$ locally for each follower.
\section{Conclusion}

In conclusion, a unified CBF-QP framework for formation tracking and network rigidity preservation of heterogeneous MAS was presented. 
Simulations confirmed that rigidity preservation is critical for accurate formation tracking under disturbances. 
\cl{Future work will extend the QP framework to 3D, incorporate obstacle avoidance, input bounds, and validate experimentally on aerial vehicles.}





\bibliographystyle{IEEEtran}
\bibliography{arxiv}

@article{saharsh2026cbfformation,
  title={Control Barrier Function only Formation Tracking in Multi-Agent Systems},
  author={Saharsh, S. and Jagtap, Pushpak},
  journal={arXiv preprint arXiv:2606.25452},
  year={2026 (accepted in IFAC World Congress 26)},
  doi={10.48550/arXiv.2606.25452 (accepted in IFAC World Congress 26)}
}

@article{franchi2012modeling,
  title={Modeling and control of UAV bearing formations with bilateral high-level steering},
  author={Franchi, Antonio and Masone, Carlo and Grabe, Volker and Ryll, Markus and B{\"u}lthoff, Heinrich H and Giordano, Paolo Robuffo},
  journal={The International Journal of Robotics Research},
  volume={31},
  number={12},
  pages={1504--1525},
  year={2012},
  publisher={SAGE publications Sage UK: London, England}
}

@article{bhatia2025decentralized,
  title={Decentralized formation control, collision avoidance and global connectivity maintenance using non-smooth barrier functions},
  author={Bhatia, Pranjal and Roy, Sayan Basu and Sujit, PB and Alvarez, Luis Mejias and McFadyen, Aaron},
  journal={Robotics and Autonomous Systems},
  pages={105272},
  year={2025},
  publisher={Elsevier}
}

@inproceedings{dias2016distributed,
  title={Distributed formation control of quadrotors under limited sensor field of view},
  author={Dias, Duarte and Lima, Pedro Urbano and Martinoli, Alcherio},
  booktitle={International Conference on Autonomous Agents \& Multiagent Systems},
  pages={1087--1095},
  year={2016}
}

@article{zelazo2015decentralized,
  title={Decentralized rigidity maintenance control with range measurements for multi-robot systems},
  author={Zelazo, Daniel and Franchi, Antonio and B{\"u}lthoff, Heinrich H and Robuffo Giordano, Paolo},
  journal={The International Journal of Robotics Research},
  volume={34},
  number={1},
  pages={105--128},
  year={2015},
  publisher={SAGE Publications Sage UK: London, England}
}

@article{sawarkar2026sliding,
  title={Sliding Mode Control for Safe Trajectory Tracking with Moving Obstacles Avoidance: Experimental Validation on Planar Robots},
  author={Sawarkar, Shubham and Sangeerth, P and Saharsh, S and Jagtap, Pushpak},
  journal={in Proceedings of 18th International Workshop on Variable Structure Systems},
  year={2026}
}

@book{liberzon2003switching,
  title={Switching in systems and control},
  author={Liberzon, Daniel},
  booktitle={Switching in systems and control},
  volume={190, 2003},
  year={Springer}
}

@inproceedings{ames2019control,

    title={Control barrier functions: Theory and applications},

    author={Ames, Aaron D and Coogan, Samuel and Egerstedt, Magnus and Notomista, Gennaro and Sreenath, Koushil and Tabuada, Paulo},

    booktitle={18th European control conference (ECC)},

    pages={3420--3431},

    year={2019}
}

@article{zhao2019bearing,
  title={Bearing-only formation tracking control of multiagent systems},
  author={Zhao, Shiyu and Li, Zhenhong and Ding, Zhengtao},
  journal={IEEE Transactions on Automatic Control},
  volume={64},
  number={11},
  pages={4541--4554},
  year={2019}
}

@article{zhao2015bearing,
  title={Bearing rigidity and almost global bearing-only formation stabilization},
  author={Zhao, Shiyu and Zelazo, Daniel},
  journal={IEEE Transactions on Automatic Control},
  volume={61},
  number={5},
  pages={1255--1268},
  year={2015}
}

@inproceedings{rayabagi2024formation,
  title={Formation control of double integrators over directed graphs using bearings and bearing rates},
  author={Rayabagi, Susmitha T and Pal, Debasattam and Mukherjee, Dwaipayan},
  booktitle={IEEE 63rd Conference on Decision and Control (CDC)},
  pages={6560--6565},
  year={2024}
}

@article{anderson2003operations,
  title={Operations on Rigid Formations of Autonomous Agents},
  author={Eren, Tolga and Anderson, Brian and Morse, A Stephen and Whiteley, Walter and Belhumeur, Peter N},
  journal={Communications in Information and Systems},
  volume={3},
  number={3},
  pages={223--258},
  year={2003}
}

@article{wu2023quadratic,
  title={Quadratic programming for continuous control of safety-critical multiagent systems under uncertainty},
  author={Wu, Si and Liu, Tengfei and Egerstedt, Magnus and Jiang, Zhong-Ping},
  journal={IEEE Transactions on Automatic Control},
  volume={68},
  number={11},
  pages={6664--6679},
  year={2023},
  publisher={IEEE}
}

@article{zhao2016localizability,
  title={Localizability and distributed protocols for bearing-based network localization in arbitrary dimensions},
  author={Zhao, Shiyu and Zelazo, Daniel},
  journal={Automatica},
  volume={69},
  pages={334--341},
  year={2016},
  publisher={Elsevier}
}

@inproceedings{pampatwar2021planar,
  title={Planar bearing-only formation control of heterogeneous multi-agent systems},
  author={Pampatwar, Akshay and Mukherjee, Dwaipayan},
  booktitle={2021 Seventh Indian Control Conference (ICC)},
  pages={171--176},
  year={2021},
  organization={IEEE}
}

\end{document}